\documentclass[12pt,a4paper]{amsart}

\usepackage{amsmath,amsthm,amssymb}
\usepackage{graphicx,wrapfig}
\usepackage{enumerate}

\def\ZZ{{\mathbb Z}}
\def\oo{\infty}

\def\g{\gamma}
\def\G{\Gamma}
\def\a{\alpha}
\def\b{\beta}
\def\d{\delta}
\def\l{\lambda}
\def\es{\emptyset}
\renewcommand{\k}{\kappa}

\newcommand{\one}{\hbox{\rm 1\kern-.27em I}}

\newtheorem{theorem}{Theorem}[section]

\newtheorem{proposition}[theorem]{Proposition}
\theoremstyle{definition}
\newtheorem{definition}[theorem]{Definition}
\newtheorem{example}[theorem]{Example}
\newtheorem{remark}[theorem]{Remark}

\author[Bj\"ornberg]{J. E. Bj\"ornberg}
\author[Britton]{T. Britton}
\author[Broman]{E. I. Broman}
\thanks{Corresponding author: E. I. Broman, Dept.\ of Mathematics,
Uppsala University, Box 256, 751 05 Uppsala, Sweden, 
Phone 0046737320791,
e-mail: broman@math.uu.se}
\author[Natan]{E. Natan}
\date{\today}

\title{A stochastic model for virus growth in a cell population}

\begin{document}
\maketitle

\begin{abstract}
This work introduces a stochastic model for the spread 
of a virus in a cell population where the virus has two ways of 
spreading: either by allowing
its host cell to live on and duplicate, or else by multiplying in
large numbers within the host cell, causing the host cell to
burst thereby letting the viruses  enter new uninfected 
cells. The model is a kind of interacting
Markov branching process.  We focus in particular the probability 
that the virus population survives and how this depends on a 
certain parameter $\l$ which quantifies the `aggressiveness'
of the virus.

Our main goal is to determine the
optimal balance between aggressive growth and long-term
success.  Our analysis shows that the
optimal strategy of the virus (in terms of survival) is obtained when
the virus has no effect on the host cell's life-cycle,
corresponding to $\l=0$.  This is
in agreement with experimental data about real viruses.
\end{abstract}

\noindent {\em Keywords:} Branching Processes, Interacting Branching Processes, Model for Virus Growth \\

\noindent {\em Subject Classification} Primary: 60J80, 60J85 Secondary:  60J27, 60J28, 92D15

\section{Introduction}
\label{intro_sec}

A virus is a simple parasitic
organism consisting of compacted genetic material in a
protein or lipid vessel.  Viruses prey on living cells,
such as bacterial or human cells, by penetrating the
membrane of the cell and transferring their genetic material into the host.
In order to multiply, the virus has two basic possibilities.
The first option 
is for the virus to temporarily
incorporate its genetic
material in the host genome, and thereby be passively
replicated along with the latter.  
The other option 
is to seize the host's replication machinery
and aggressively replicate, thereafter releasing its progeny
in the surrounding medium.  The
`free virions' must then attach
to new host cells within a short time in order to survive.
For many viruses, this process necessarily involves
bursting the host cell, thereby killing it.

The technical term for the event that a virus
bursts its host cell is \emph{lysis}, and one says
that the virus \emph{lyses} the host cell.  A virus which
is incorporated into, and passivlely replicated along with,
the host genome is said to be in the \emph{lysogenic state},
or to employ the \emph{lysogenic strategy}.
Sometimes one speaks loosely of `the lytic strategy'
to denote that a virus `becomes
lytic', that is to say
actively lyses the host cell.  
A lysogenic virus will eventuelly become lytic;
in fact, it is well-established
experimentally~\cite{lwoff53} 
that viruses in the lysogenic state
will revert to the lytic strategy if the host cell
is under stress and in danger of dying, 
enabling the virus to find a `safer' host.

We introduce a stochastic model to investigate this behaviour.  The
model is a two-dimensional Markov process $(X(t),Y(t))_{t\geq 0}$, where
$X(t)$ is the number of `healthy cells' at time $t$, and
$Y(t)$ is the number of `infected cells' (i.e.\ cells
having virus in them). Both components 
$(X(t))_{t\geq0}$  and $(Y(t))_{t\geq0}$ behave in many ways like
branching processes, although there are dependencies between them.
A healthy cell is replaced by a random number of new healthy
cells at rate 1.  This random number is independent of
other events and drawn from a distribution
$(p_k)_{k\geq0}$.  Infected cells behave similarly, although they
are replaced by $k$ new cells at rate $p_k$ if $k\geq 1,$ 
while they are replaced by 0
new cells (die) at the higher rate $p_0+\lambda$. 
Here $\lambda\geq 0$ is a parameter
that reflects the negative impact of the virus on the hosts
lifelength.  When an infected cell dies (i.e.\! is replaced
by 0 new cells), it bursts (lyses) and releases `free virions'.
These free virions immediately enter a random number of 
healthy cells, thus converting them into infected cells. 
The number of new infections is independent of all other
events, and is drawn from a distribution $(\g_k)_{k\geq0}$.
The model is defined in detail in
Section~\ref{def_sec}.

We are concerned with a fundamental question about the
virus' reproductive strategy, namely:  what is the optimal
level of `aggressiveness' (balance of lysis to lysogeny)
from the point of view of the virus?
Here we interpret `optimality' as maximizing the chance of
the virus establishing itself in the cell population and, 
ultimatly, surviving in the long-term.
We therefore study the \emph{extinction probability} $\eta$ of the
infected process $(Y(t))_{t\geq 0}$
(see Definition~\ref{eta_def}).  We are interested in $\eta$, 
or rather $1-\eta$, as an indicator
of the `fitness' of the virus, and are mainly
concerned with how it depends on $\l$.  This is because $\l$
governs the relative rate of lysis events, and is thus
a measure of 
the `level of aggressiveness' of the virus.

For the experimentally well-studied virus Lambda,
the lysogenic state appears overwhelmingly stable.
Once in the lysogenic (dormant) state, it has been
found very unlikely to spontaneously switch to the lytic
state~\cite{aurell_sneppen,little99}:
a spontaneous
transition to the lytic state occurs about once
in $10^7$
generations~\cite{aurell_sneppen}.
This is lower than the mutation rate of the 
incorporated viral
genome, which is once in $10^6$ to $10^7$
generations~\cite{little99}.
It is natural to
ask if this lysogenic stability 
is an advantage for the success of the virus infection?
For the virus Lambda, a choice between lytic and lysogenic
also occurs \emph{at} the moment of infection.
We focus mainly on the virus'
decision \emph{after} it has been incorporated in the
host genome, but in Section~\ref{imm_sec} 
briefly deal also with the decision at the
moment of infection.

Our model is of course a simplification of real virus populations.
Indeed, we make the following basic simplifications: life-lengths of
cells are assumed to be independent and exponentially distributed,
spatial separation and locations of cells are not taken into account,
and the number of new infections caused by one infected cell is
assumed not to vary with the population sizes $(X(t))_{t\geq 0}$ and
$(Y(t))_{t\geq 0}$ (except if \emph{all} remaining healthy cells are
to become infected).
The advantage of making such simplifications is that a detailed and
rigorous analysis can be performed, hopefully
highlighting general principles that can then form the basis
for more realistic modelling.  

It is well-known \cite{athreya_ney,haccou_jagers} that a
branching process either dies out, or grows exponentially fast for all
time.  Thus a branching process 
is not a realistic long-term model for population size, in light
of the limited resources in the real world.  
Instead, we see the surivival probability $1-\eta$
as an indicator of
the probability that a virus population establishes itself
in a population of healthy cells in the first place. 
In this sense our model is primarily relevant for the
early stages of a virus infection and the competition between
two growing populations. 
Since our model is concerned with qualitative
properties of reproductive strategies, not with numerical estimates of
population size, we do not see the use of branching processes as a
limitation.  In what follows we will refer to $1-\eta$
as the `survival probability' as this is the appropriate
term in the context of branching process theory, bearing in
mind that when interpreting our results in terms of real
viruses one should rather think of $1-\eta$ as an
indicator of the relative success in establishment and
proliferation.

Recall the concept of stochastic ordering of probability vectors:
if $\pi=(\pi_k)_{k\geq0}$ and $\pi'=(\pi'_k)_{k\geq0}$ are
probability vectors then we say that $\pi'$ is
\emph{stochastically larger} than $\pi$ if
\[
\sum_{j\geq k}\pi'_j\geq \sum_{j\geq k}\pi_j
\]
for all $k\geq 0$.  We denote this by $\pi\preceq \pi'$.
The following is the main result of this paper.
\begin{theorem}\label{mono_thm}
For $\l\geq0,$ $\g_0=0$ and any starting conditions $X(0)\geq 1,Y(0)\geq 1,$ 
we have that $\eta=\eta(\l,(\g_k)_{k\geq0})$
is monotonically increasing in $\l$ and $(\g_k)_{k\geq0}$.
\end{theorem}

Thus, loosely speaking, the virus maximizes its
survival probability by being as \emph{passive} as possible 
(i.e. when $\lambda=0$). 
This is in agreement with the observed stability of the
lysogenic state for real viruses, see 
Section~\ref{disc_sec}.
However, the full details of how the `fitness' $1-\eta$
depends on $\l$ are complex, and depend on the other parameters
 of the process.  
Furthermore, simulations and heuristic
arguments suggest that 
monotonicity in $\l$ may hold under weaker assumptions
(Example~\ref{mono_ex} and Section~\ref{gc_sec}) than in Theorem \ref{mono_thm}, 
but interestingly $\eta$ is
\emph{not} monotone in $\l$ for all choices of the other
parameters (Proposition~\ref{exprop}).  We have not
been able to find
a counterexample to optimaly of $\l=0$.

Thus this paper highlights the principle 
that in order to achieve long-term survival it may be better to `be kind'
to your host environment, even if this hampers your short-term expansion.
We do not aim to give a complete and final picture, however, and there
are many interesting questions and research directions that fall outside
the scope of the current paper. For instance, in \cite{BB}, the rigorous
mathematical treatment of the model will be continued as will be explained
in more detailed in later sections. 
Other possible directions include studying
real life data, and in the cases when a rigorous mathematical treatment is
unfeasible, use simulation studies to connect the proposed model to these data.

Most previous models studying the spread of viruses are non-stochastic
and formulated in terms of differential equations.  A notable example
is that of~\cite{nowak_may}.   We prefer to
formulate our model in microscopic terms, deducing macroscopic
properties explicitly from our assumptions about the interactions of
the particles involved. To our knowledge the current model has not
been studied before but stochastic models of similar `nature' appear
for example in predator-prey 
models~\cite{renshaw} and epidemic models, in particular
models for competing 
epidemics~\cite{kendall_saunders}.

\vspace{.3cm}
\noindent{\bf Acknowledgements}:
The authors would like to thank Leonid Hanin for 
helpful comments on a draft of this article.

\section{The model}

\subsection{Definition}\label{def_sec}
Let $(p_k)_{k\geq0}$ and $(\g_k)_{k\geq0}$ be probability
distributions on the nonnegative integers, and let $\l\geq0$.
We assume throughout that the means $\sum_{k\geq0}kp_k$ and
$\sum_{k\geq0}k\g_k$ are finite.
We exclude the (degenerate) case when $p_1=1$;  in fact the reader
may for convenience assume that $p_1=0$, since this only amounts to a 
time-change.

The continuous--time Markov chain $(X(t),Y(t))_{t\geq 0}$, 
taking values in $\ZZ^2_+$, was  informally described
in Section~\ref{intro_sec}.  To recapitulate the main points,
each healthy cell is replaced by $k\geq0$ new healthy cells
at rate $p_k$.  Being replaced by $k=0$ new cells corresponds
to dying.  Each infected cell is replaced  by $k\geq1$
new infected cells at rate $p_k$.  When an infected cell
dies, which occurs at rate $p_0+\l$, a random number
of healthy cells are converted into infected cells.
If $t$ is the time of such an event,
we draw a random variable $\Gamma_t$ from the distribution $(\g_k)_{k\geq0}$
independently of other events. If $\Gamma_t\leq X(t)$ we simply declare 
$\Gamma_t$ of previously healthy cells to be infected, while if
$\Gamma_t > X(t)$ we declare all previously healthy cells to be infected.
To define this process formally, we list the 
different possible jumps in Table~\ref{rate_tab}.
\begin{table}[hbt]
\centering
\begin{tabular}{l|l|l|l}
& Transition & Rate \\ \hline
(i) & For $k\geq 0$: & \\
 & $(x,y)\rightarrow (x+k-1,y)$ & $x p_k$\\
(ii) & For $k\geq 1$: & \\
 & $(x,y)\rightarrow (x,y+k-1)$ & $y p_k$\\
(iii) & For $k\geq 0$: & \\
& $(x,y)\rightarrow (x-(x\wedge k)),y-1+(x\wedge k))$ &
$y(p_0+\l)\g_k$
\end{tabular}
\caption{Transition rates for the process $(X(t),Y(t))_{t\geq0}$,
  valid for $x,y\geq0$.}
\label{rate_tab}
\end{table}

Note that, for certain combinations of $x$, $y$ and $k$,
the same transition occurs mutliple times in Table~\ref{rate_tab}.
The correct interpretation is to add the corresponding rates.  For 
example, the transition $(0,y)\mapsto(0,y-1)$
occurs at rate $y(p_0+\l)=\sum_{k\geq 0}y(p_0+\l)\g_k$.

To avoid trivial cases, we assume
throughout that $X(0),Y(0)\geq 1.$  Biologically it might
be most relevant to consider the case when $p_k=0$ for
$k\geq 3$, but none of our results depend on any special
assumptions about $(p_k)_{k\geq0}$ so we will consider
general distributions.

We now state some immediate properties of the model.  If it were the
case that $Y(t)=0$, then healthy cells would evolve as a Markov
branching process, with intensity $1$ and offspring distribution
$(p_k)_{k\geq0}$. 
Similarly, if $X(t)=0$ for some $t,$ then $(Y(t+s))_{s\geq 0}$ would
behave like a Markov branching process with the higher intensity
$(1+\lambda)$ and an offspring distribution $(p'_k)_{k\geq0}$
derived from $(p_k)_{k\geq 0}$ by placing more mass
on $k=0$ (see~\eqref{p'_eq} below).  When
both components are positive, as transition rate (iii) tells us, then
healthy cells may turn into infected cells. This scenario hence
`helps' the process $(Y(t))_{t\geq 0}$ and 
`hurts' the process $(X(t))_{t\geq 0}$.

Note that 
the virus is assumed not to change the offspring distribution of
\emph{surviving} cells.  This, together with the
increased mortality rate of infected cells, determines the
form of the transition rates above, see~\eqref{xi'_eq}.
Also                            
note that the random number drawn from the distribution $(\g_k)_{k\geq 0}$
is the number of \emph{new infections} due to a lysis event,
rather
than the number of `free virions'.
In the present work we consider only this
simplified formulation, leaving more realistic modifications
for future work.

\subsection{The extinction probability}
\label{ext_sec}

As explained in Section~\ref{intro_sec},
we view the extinction probability of the process $(Y(t))_{t\geq 0}$
as an indicator of the fitness of the virus:
\begin{definition}\label{eta_def}
Let
$\eta=\lim_{t\rightarrow\oo}P(T(t)=0)=P(Y(t)\rightarrow0)$ denote the extinction probability of
the process $(Y(t))_{t\geq 0}$.
\end{definition}
Thus `small' $\eta$ corresponds to `high fitness'.
Note that $\eta$ is a function of the parameters
$(p_k)_{k\geq0}$, $(\g_k)_{k\geq0}$, $\l$, $X(0)$ and $Y(0)$.  
For the reasons given in Section~\ref{intro_sec} we are
mainly interested in how $\eta$ depends on $\l$
and the distribution $(\g_k)_{k\geq0}$.

\begin{remark}
From a purely mathematical point of view the allowed range of
values of $\l$ is $\l\geq-p_0$.  It is not
hard to see that $\eta(-p_0)=0$ and $\eta(\lambda)\to 1$ as
$\lambda\to\infty$.
However, viruses being parasites, from a
biological point of view it seems unlikely that  infected cells
should have \emph{longer} life-length than healthy
cells.  Throughout the
rest of the paper we will therefore assume that $\lambda \geq 0$.
\end{remark}

\section{Additional results} \label{addres_sec}


Let $\xi$ be a random variable with
distribution $(p_k)_{k\geq0}$.  As mentioned above, if $X(t)=0$ for some $t,$
then from that time onwards, the process $(Y(t+s))_{s\geq 0}$ is a standard
branching process.  Its intensity is then $1+\l$ and
its offspring distribution $(p'_k)_{k\geq0}$ is given by:
\begin{equation}\label{p'_eq}
p'_0=(p_0+\l)/(1+\l) \mbox{ and }
p'_k=p_k/(1+\l) \mbox{ for } k\geq 1.
\end{equation}  
Let $\xi'$
be a  random variable with distribution $(p'_k)_{k\geq0}$, and note that
for $\l>0$, $p'_0=P(\xi'=0)>P(\xi=0)=p_0$, whereas
\begin{equation}\label{xi'_eq}
P(\xi'=k\mid\xi'\neq0)=P(\xi=k\mid\xi\neq0) \mbox{ for all } k\geq 1.
\end{equation}
This choice of $(p'_k)_{k\geq0}$ 
is the only one, given $\lambda,$ such that
the intensity at which a cell gives birth 
to $k\geq 1$ new cells is the same for
both $(X(t))_{t\geq 0}$ and $(Y(t))_{t\geq 0}.$

Let $\G$ be a random variable independent of $\xi'$,
with distribution $(\g_k)_{k\geq0}$.  Write
\begin{equation} \label{eqn1}
\psi=\xi'+\G\cdot\one\{\xi'=0\}.
\end{equation}
Then $\psi$ has distribution $(q_k)_{k\geq0}$, where
\begin{equation}\label{q_eq}
q_0=\frac{\g_0(p_0+\l)}{1+\l},\mbox{ and }
q_k=\frac{p_k+\g_k(p_0+\l)}{1+\l}\mbox{ for }k\geq 1.
\end{equation}
Write
\begin{equation}\label{t_def}
T_X:=\inf\{t\geq0:X(t)=0\}
\end{equation}
for the (possibly infinite) time when the healthy population
becomes extinct.
The following summarizes some of the previous discussion:
\begin{proposition}\label{br_prop}
\hspace{1cm}
\begin{enumerate}
\item If $Y(t)=0$ for some $t\geq0$, 
 then $(X(t+s))_{s\geq 0}$
is a Markov branching process with
intensity 1 and offspring distribution $(p_k)_{k\geq0}$;
\item If $X(t)=0$  for some $t\geq0$,
then $(Y(t+s))_{s\geq 0}$
is a Markov branching process with
intensity $1+\l$ and offspring distribution $(p'_k)_{k\geq0}$;
\item The process $(Y(t))_{0\leq t<T_X}$ is a
(stopped) Markov branching process with
intensity $1+\l$ and offspring distribution $(q_k)_{k\geq0}$.
\end{enumerate}
\end{proposition}

For proofs of the following basic facts about branching
processes, see~\cite{athreya_ney,haccou_jagers,harris63}.
Consider an arbitrary Markov branching process $(W(t))_{t\geq 0}$ 
with lifelength intensity $a$, and offspring distribution 
$(z_k)_{k\geq0}$ such that the mean $\sum_{k\geq0}kz_k$ is finite.
Let $Z$ be a random variable
with distribution $(z_k)_{k\geq 0}$.  The number
\[
\mu=a\cdot (E(Z)-1),
\]
is called the
\emph{Malthusian parameter} of the process $(W(t))_{t\geq 0}$.
Let $A=\big\{W(t)=0\mbox{ for some }t\geq 0\big\}$ be the event 
of extinction.  It is well-known that $P(A)=1$ if $\mu\leq0$.
If $\mu>0$ then 
\[
\lim_{t\rightarrow\oo}\frac{\log W(t)}{t}=\mu,
\mbox{ almost surely on the complement $A^c$}.
\]

Write $\a$ and $\b$ for the Malthusian parameters of  branching
processes with respective  intensities $1$ and $1+\l$, and
offspring distributions $(p_k)_{k\geq0}$ and $(q_k)_{k\geq0}$, as in parts (1) and (3)
of Proposition~\ref{br_prop}.  Thus $\a$ is the parameter for the
uninfected population in the absence of infected cells, and $\b$ for
the infected population in the presence of a very large uninfected
population.  We have that
\begin{equation}
\a=E(\xi) -1,
\qquad
\b=(1+\l)(E(\psi)-1), \label{alpha_beta}
\end{equation}
where $\xi$ and $\psi$ are as above.
Using (\ref{eqn1}), we find that
\begin{equation}\label{beta_alpha}
\b=\a+p_0E(\G)+\l(E(\G)-1).
\end{equation}

\begin{proposition}\label{takeover_prop}
Suppose $p_0>0$ or $\l>0$.
Then $P(T_X<\infty)=1$ if and only if either
$\a\leq0$ or $\g_0=0$.
\end{proposition}
We do not prove this result in detail here, but note that the
sufficiency of the condition $\a\leq0$ is immediate from the
properties of branching processes described above.  If $\a>0$,
the necessity of the condition $\g_0=0$ is immediate, since if
$\g_0>0$ there is positive chance that $Y(t)=0$ at the time
of the first transition, while the healthy process survives.
Intuitively, the sufficiency of the condition $\g_0=0$ follows from
the fact that $Y(t)$ is `immortal' as long as $X(t)\neq0$
and (from~\eqref{beta_alpha}, since $E(\G)>1$)
grows \emph{much} faster
than $(X(t))_{t\geq 0}$, meaning that there will be very many infection
events for large $t$.  It is not difficult to make this intuition
rigorous, in fact this will be proved in the upcoming paper \cite{BB}.

Recall the main result Theorem \ref{mono_thm}.
In words, the assumption $\g_0=0$ says that a lysis event
always leads to new infections;  thus the failure rate
of infections is zero.
The following proposition
shows that we cannot remove the condition
$\g_0=0$ from Theorem \ref{mono_thm}, and still come to the
same conclusion. We will address this
further in Section~\ref{disc_sec}.
\begin{proposition} \label{exprop} \hspace{1mm}
\begin{enumerate}
\item There exist $(p_k)_{k\geq0}$, $\l$ and probability
vectors $\g^{(1)}\preceq\g^{(2)}$ so that
$\eta(\l,\g^{(1)})>\eta(\l,\g^{(2)})$.
\item Furthermore, there exist $(p_k)_{k\geq0},(\g_k)_{k\geq0}$ 
and $\l_1<\l_2$ so that
$\eta(\l_1)>\eta(\l_2)$.
\end{enumerate}
\end{proposition}

\begin{figure}[hb]
\centering
 \includegraphics[width=1\textwidth]{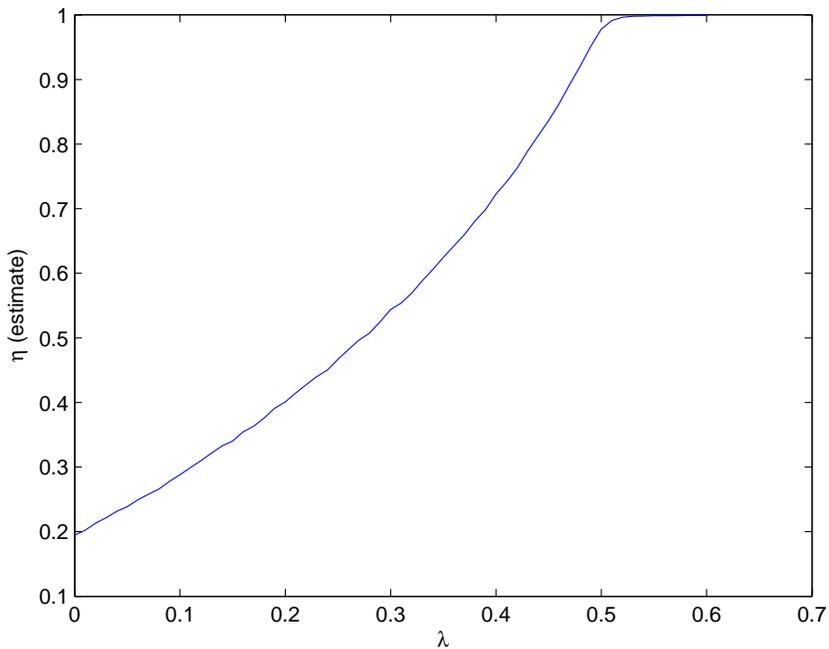}
\caption{Illustration of Example \ref{mono_ex}.}
\label{bjollaeettpucko}
\end{figure}

Proposition~\ref{exprop} is proved in
Section~\ref{proof_sec}.
Our proof of the second part requires taking
$E(\G)<1$.  It is natural to
guess that the condition $\g_0=0$ in Theorem~\ref{mono_thm}
can be replaced by the condition $E(\G)\geq1$:
infections are successful `on average'.  This is still
an open problem, but the following simulation supports
this guess.



\begin{example}\label{mono_ex}
Figure \ref{bjollaeettpucko}
shows estimated values for $\eta(\l)$
when $p_0=1/4$, $p_2=3/4$, $\g_0=9/20$ and $\g_2=11/20$.
Note that $E(\G)>1$;  the simulation suggests that $\eta$
is increasing in $\l$.  The estimates were obtained by
running, for each value of $\l$, the process $10^5$ times,
for $2\cdot 10^4$ transitions each.  The estimate
for $\eta$ is the fraction of runs where $Y$ did not
die out.  (Usually $Y$ was either 0 or very large
at the end of a run.)   For reasons similar to
the second part of Proposition~\ref{exprop},
the true value of 
$\eta(\l)$ for $\l\geq0.5$ is 1.  However for $\l=0.5$
the process $(Y(T_X+s))_{s\geq0}$ is a critical
branching process so the time until extinction is large,
accounting for the small deviation from 1 in the graph.
\end{example}

\subsection{Initial conditions and connection to ODE}
\label{init_sec}

Here is a brief informal discussion about how $\eta$
depends on the starting condition $(X(0),Y(0))$
when $\l$ and $(\g_k)_{k\geq0}$ are fixed.
We focus on two starting conditions of
potential interest, namely $(X(0),Y(0))=(n,1)$ and $(X(0),Y(0))=(n,m)$
where we think of $n$ and $m$ as
large, but $n\gg m$. The first situation is where one cell,
surrounded by healthy cells, is infected. The second situation would
correspond to the case where a large number of healthy cells are
encountered by a large number of infected cells.
Formal results would be stated asymptotically as
$m,n\rightarrow\infty$.

The `take-home message' of 
Proposition~\ref{br_prop} is that $(Y(t))_{t\geq0}$
essentially behaves like a branching process,
with a Malthusian parameter that depends on whether
$X(t)=0$ or $X(t)>0$.  If $X(t)>0$ the Malthusian
paramter is 
$\b=\a+p_0 E(\G)+\l(E(\G)-1)$
as already stated in~\eqref{beta_alpha};  if $X(t)=0$
the Malthusian parameter is $\b'=\a-\l$,
as is easily deduced from the second part 
of Proposition~\ref{br_prop}.  Here $\a=E(\xi)-1$
is as given in~\eqref{alpha_beta}.  Note that $\b'<\b$.

The most interesting behaviour occurs if $\a,\b>0$,
which we assume henceforth.  We also assume that
$\g_0>0$, since the case $\g_0=0$ is easily analyzed
using Proposition~\ref{takeover_prop}
and Theorem~\ref{mono_thm}.
The main qualitative differences in behaviour occur
according as $\b>\a$ or $\b<\a$;  in the former case
there are two interesting subcases, namely $\b'\leq0$ and $\b'>0$.

Let us start by considering the initial condition 
$X(0)=n$, $Y(0)=m$.
Since $Y(0)=m$ is large and $\b>0$, most likely the infected
population starts growing.  Roughly speaking, $Y(t)\approx me^{\b t}$.
Suppose first $\b>\a$.  The healthy population $X(t)$ will 
then at most
be of order $ne^{\a t}\ll me^{\b t}$, and  eventually $Y(t)$ will so far
exceed $X(t)$ that infections will overwhelm the healthy
population, so that we get $X(t_0)=0$ for some time $t_0$.  From then
on the infected population will have Malthusian parameter
$\b'$.  If $\b'\leq 0$ this means that $\eta$ will
equal 1, whereas if $\b'>0$ then $\eta$ will be 
close to 0, since $Y(t_0)$ is large.  On the other hand, if 
$\a>\b$ then typically $X(t)$ will be so much larger
than $Y(t)$ that the healthy population does not `feel'
the presence of the infection.  Since $\b>0$ and $Y(0)=m$
is large it follows that $\eta$ is close to 0.

Now consider the case $(X(0),Y(0))=(n,1)$.  This is similar to the
case $(X(0),Y(0))=(n,m)$, except that there is a considerable
chance (probability at least $\g_0(p_0+\l)/(1+\l)$, this being the
probability of a lysis leading to no infections)
that the infected process dies out in short time.
However, if $Y(t)$ \emph{does} start growing then its
size will eventually be in the order of $e^{\b t}$.
From then on the same intuition as for the starting
condition $(n,m)$ is valid.

Another way to understand the behaviour described above,
in particular the starting condition $(n,m)$,
is to look at the ODE model corresponding to our model.
Letting $\d=(p_0+\l)E(\G)$, this is given by 
\[
\dot{x}(t)=\a x(t)-\d y(t),\quad \dot{y}(t)=\b y(t),
\]
as long as $x(t)>0$.  If $\a\neq \b$
the solution is 
\[
y(t)=y(0)e^{\b t},\quad
x(t)=\Big(x(0)+y(0)\frac{\d}{\b-\a}\Big)e^{\a t}
-y(0)\frac{\d}{\b-\a}e^{\b t}.
\]
It is not hard to see that $x(t)$ eventually reaches
the absorbing state 0 if
$\b>\a$, whereas $x(t)\rightarrow+\oo$ if $\a>\b$
and $x(0)/y(0)=n/m$ is sufficiently large,
corresponding to whether the infection takes over or not. 
We will not study the ODE model further
since it cannot give information about survival
probabilities.

\section{Proofs}\label{proof_sec}

In this section we prove Theorem~\ref{mono_thm}
and Proposition~\ref{exprop}.
In what follows we will work with two processes
$(Z_i(t))_{t\geq 0}=(X_i(t),Y_i(t))_{t\geq 0}$, with parameters 
$(p_k)_{k\geq0}$, $\l_i$ and $(\g^{(i)})_{k\geq0}$, $i=1,2$,
respectively.  We  write $T_i=\inf\{t\geq0:X_i(t)=0\}$.

\begin{theorem}\label{coupling_thm}
Let $0\leq\l_1\leq\l_2$,
$\g^{(1)}\preceq\g^{(2)}$, and 
$\g^{(1)}_0=\g^{(2)}_0=0$.  Then there is
a coupling of the processes $(Z_i(t))_{t\geq 0}$
($i=1,2$) such that the following hold almost surely:
\begin{enumerate}
\item $T_2\leq T_1$,
\item $Y_2(t)\geq Y_1(t)$ for all $t<T_2$,
\item $X_2(t)\leq X_1(t)$ for all $t\leq T_2$, and
\item $X_1(t)+Y_1(t)\geq X_2(t)+Y_2(t)$ for all $t\geq 0$.
\end{enumerate}
\end{theorem}

Before we prove Theorem \ref{coupling_thm},
we show how Theorem~\ref{mono_thm} follows, almost immediately,
from Theorem~\ref{coupling_thm} and Proposition~\ref{takeover_prop}.

\begin{proof}[Proof of Theorem~\ref{mono_thm}]
By Proposition~\ref{takeover_prop} we have $T_1<\infty$
almost surely.  For $t>T_1$ it follows from parts (1)
and (4) of Theorem~\ref{coupling_thm} that
$Y_1(t)\geq Y_2(t)$, so that
$P(Y_1(t)\rightarrow0)\leq P(Y_2(t)\rightarrow0)$.
\end{proof}

\begin{proof}[Proof of Theorem~\ref{coupling_thm}]
The basic strategy is to `twin' cells in the process
$(Z_1(t))_{t\geq 0}$ with cells in the process
$(Z_2(t))_{t\geq 0}$ so that `events' in one process
correspond with `events' in the other process.
Intuitively, the reason we can acheive a coupling as the
one claimed is that only `lysis events' occur at a higher
rate in the second process:  these events always increase the
infected population, but decrease both the healthy and
total populations.  Here are the details.

It will be convenient to think of
$Z_i(t)$, $X_i(t)$, and $Y_i(t)$ ($i=1,2$) as sets of individual cells.
Formally, we could label the elements in the set $Z_i(t)$ by $(x_1,i),\ldots
(x_{|X_i(t)|},i),(y_1,i),\ldots, (y_{|Y_i(t)|},i).$ However, this notation 
would quickly become cumbersome, and so we will use a somewhat less formal,
although still rigorous, approach.
We will describe the transitions of the
coupled process $(Z_1(t),Z_2(t))_{t\geq 0}$ at the level of \emph{pairs}
$(a,b)$ of individual cells, where $a\in Z_1(t)$ and $b\in Z_2(t)$.
For each $t\geq0$, each element of $Z_1(t)\cup Z_2(t)$ is
required to belong to a unique such pair, and we say that $a$
and $b$ are \emph{twinned} if they belong to the same pair.
We allow the possibility $b=\es$, in which case we say that
$a$ is \emph{untwinned} (the possibility $a=\es$ will not occur).
We will say that a cell is of \emph{type~1}
(respectively, \emph{type~2}) if it belongs to $Z_1(t)$
(respectively,~$Z_2(t)$).

It is a standard consequence of
the ordering $\g^{(1)}\preceq \g^{(2)}$
that we may couple two random variables
$\G^{(1)}$ and $\G^{(2)}$ such that $\G^{(1)}$ has law $\g^{(1)}$
and $\G^{(2)}$ has law $\g^{(2)}$ and $P(\G^{(1)}\leq\G^{(2)})=1$.
We assume henceforth that $\{(\G_{1,i},\G_{2,i})\}_{i\geq 1}$
is an i.i.d.\
sequence such that $\G_{1,i}\leq \G_{2,i}$
for every $i,$ and that $\G_{1,i}$ has law $\g^{(1)}$
and $\G_{2,i}$ has law $\g^{(2)}.$ In the construction that follows below, if there is a
lysis event at some time $\tau,$ we let $I(\tau)$ be the smallest $i$ such that
$(\G_{1,i},\G_{2,i})$ has not previously been used in the construction.

We start by describing the coupling
$(Z_1(t),Z_2(t))_{t \geq 0}$ up to time
$T_1\wedge T_2$ (it will transpire that $T_1\wedge T_2=T_2$).
It will be convenient to think of healthy cells as
coloured \emph{blue} and infected cells as coloured \emph{red}.
The allowed colour combinations before time $T_1\wedge T_2$
are the following (the right column introduces notation
for the number of pairs of each colour combination):
\begin{center}
\begin{tabular}{l|c}
Colour & Number \\ \hline
$(B,B)$ & $b_2$ \\
$(R,R)$ & $r$ \\
$(B,R)$ & $b_1$ \\
$(B,\es)$ & $b_0$
\end{tabular}
\end{center}

Now we turn to describing the transitions.
\begin{enumerate}[(i)]
\item
Any pair $(a,b)$ is replaced by $k\geq1$ new pairs
at rate $p_k$;  all the new pairs have the same
colour combination as the original pair $(a,b)$ (also in the case
$b=\es$).
\item A pair of colour $(B,B)$
or $(B,\es)$ is deleted at rate $p_0$.
\item A pair $(a,b)=(R,R)$ can give rise to the following additional transitions.

Firstly, at rate $p_0+\l_1$ it has a \emph{type-1-lysis}.
If $\tau$ is the time of such an event,
we first delete $(a,b)$, and then one of the
following cases occur.

\emph{Case 1:} $\G_{2,I(\tau)}\leq b_2$.
Then we take $\G_{1,I(\tau)}$ pairs of colour $(B,B)$
and change their colour to $(R,R)$, and we take
$\G_{2,I(\tau)}-\G_{1,I(\tau)}$ of the
remaining $(B,B)$-pairs and change their colour to $(B,R)$.

\emph{Case 2:} $\G_{1,I(\tau)}>b_2$.  Then we start by
changing all the $b_2$ pairs of colour $(B,B)$ to $(R,R)$.
Let $\G_{1,I(\tau)}':=\G_{1,I(\tau)}-b_2>0$, and
proceed by changing $\G_{1,I(\tau)}' \wedge b_1$ pairs of colour $(B,R)$
to colour $(R,R)$.
Proceed by letting $\G_{1,I(\tau)}^{''}:=(\G_{1,I(\tau)}'-b_1) \vee 0$
and changing $\G_{1,I(\tau)}^{''}\wedge b_0$ pairs
of colour $(B,\es)$ to $(R,\es)$.  Note that in this case we arrive at time
$T_2\leq T_1$, and that this inequality is strict if  and only if
there remain pairs of colour $(B,R)$ or $(B,\es)$ after the changes are made.
In particular, we only create any $(R,\es)$-pairs
if we arrive at time $T_2$.

\emph{Case 3:} $\G_{2,I(\tau)}>b_2\geq \G_{1,I(\tau)}$.
Then we first take $\G_{1,I(\tau)}$ pairs of colour $(B,B)$
and change their colour to $(R,R)$, and then take the remaining
$(b_2-\G_{1,I(\tau)})$
pairs of colour $(B,B)$ and change them to $(B,R)$.  In this case we
arrive at time $T_2<T_1$.

Secondly, at rate $\l_2-\l_1$ the pair $(a,b)$
has a \emph{type-2-lysis}.
If $\tau$ is the time of such an event:
\begin{equation}\label{t2l}
\begin{array}{ccl}
& \bullet & \text{the original pair $(a,b)$ stays unaltered;}\\
& \bullet & \text{one pair of colour $(B,B)$ gets replaced by $(B,\es)$;}\\
& \bullet & \text{$(\G_{2,I(\tau)}-1)\wedge(b_2-1)$ of
the remaining $(b_2-1)$ pairs} \\
& & \text{of colour $(B,B)$ have their colour changed to $(B,R)$.}
\end{array}
\end{equation}
Note that we arrive at time $T_2<T_1$ if $\G_{2,I(\tau)}\geq b_2$.

\item
Finally, a pair $(a,b)$ of colour $(B,R)$ can give rise to
the following additional transitions,

Firstly, at rate $p_0$ a \emph{type-1-death} occurs.
If $\tau$ is the time of such an event,
then $(a,b)$ is deleted and $\G_{2,I(\tau)}\wedge b_2$
pairs of colour $(B,B)$ are changed
to $(B,R)$.  Note that if $\G_{2,I(\tau)}\geq b_2$ then we arrive at time
$T_2<T_1$.

Secondly, at rate $\l_2$ a type-2-lysis.
This yields the exact same transition as described
in (~\ref{t2l}).
Again, if $\G_{2,I(\tau)}\geq b_2$ then we arrive at time $T_2<T_1$.
\end{enumerate}

It is straightforward to check that the described coupling produces the right
marginal dynamics. It might be helpful however to explain why we replace
a pair $(B,B)$ by $(B,\es)$ in~\eqref{t2l}.
Here, the cell that lyses is destroyed
and one might expect a transition from $(R,R)$ to $(R,\es).$ However,
since the type 2 cell that lyses infects cells of $X_2(t),$
we pick one of the newly infected type 2
cells (belonging to a $(B,B)$ pair), and twin it with the previous type 1
twin that did not undergo a lysis. The effect
is the same as replacing a pair $(B,B)$ by $(B,\es)$ and
leaving $(R,R)$ unchanged.

The description above applies until time $T_2$;  the construction implies
that $T_2\leq T_1$, proving the first part of the theorem.
Since a type 1 cell is of colour $R$ only if it is twinned with
a type 2 cell of colour $R$ (for $t<T_2$),
the second part of the theorem also follows.
Similarly, a type 2 cell is coloured $B$ only if it is
paired with a type 1 cell of colour $B$
(for $t\leq T_2$), proving the third part
of the theorem.

The construction so far also implies the final part of the
theorem for the range $0\leq t\leq T_2$ since every type 2 cell has
a twin;  the construction for $t>T_2$, which we will describe now,
will preserve this property.

For $t\geq T_2$ the process $(Z_2(t))_{t \geq T_2}$ consists only of infected cells.
It will be convenient now to think of the cells of type 2 as
\emph{green}. The reason is that in the absence of healthy cells,
the process $(Y_2(t))_{t \geq T_2}$ evolves differently compared to 
when $t<T_2,$
as described above and in Proposition~\ref{br_prop}. Therefore,
we keep the colours blue and red for healthy and
infected cells of type 1, respectively.  At a time $t\geq T_2$ we
then have the following possible colour combinations
(the right column introduces notation
for the number of pairs of each colour combination, note that we are redefining
$b_0$ and $b_1$):
\begin{center}
\begin{tabular}{l|c}
Colour & Number \\ \hline
$(B,\es)$ & $b_0$ \\
$(B,G)$ & $b_1$ \\
$(R,G)$ & $r_1$ \\
$(R,\es)$ & $r_0$
\end{tabular}
\end{center}
The following transitions may occur.
\begin{enumerate}[(i)]
\item As before, any pair $(a,b)$ is replaced by $k\geq 2$
identical pairs at rate $p_k$.
\item A pair of colour $(B,\es)$ or $(B,G)$ is deleted at rate $p_0$.
\item A pair of colour $(B,G)$ is, additionally, changed to $(B,\es)$
at rate $\l_2$.
\item A pair $(a,b)$ of colour $(R,\es)$
or $(R,G)$ lyses at rate
$p_0+\l_1$. If $\tau$ is the time of such an event, then
we delete $(a,b)$, change the
colour of $\G_{1,I(\tau)}'=\G_{1,I(\tau)}\wedge b_0$ pairs of colour $(B,\es)$ to $(R,\es)$,
and finally change the colour of $(\G_{1,I(\tau)}-\G_{1,I(\tau)}')\wedge b_1$ pairs of colour
$(B,G)$ to $(R,G)$.
\item Additionally, a pair of colour $(R,G)$ is replaced by a pair
$(R,\es)$ at rate $\l_2-\l_1$.
\end{enumerate}
Since a green cell is always twinned with a type 1 cell,
this establishes the result.

As before, it is elementary to check that the described coupling produces the right marginal dynamics.
\end{proof}

\begin{proof}[Proof of Proposition~\ref{exprop}]
For the first part, consider the models where $p_2=1$, $0<\l<1$ is
arbitrary, and where $\g^{(1)}_0=1$ and $\g^{(2)}_1=1$.  Clearly
$\g^{(1)}\preceq\g^{(2)}$.  Write $(X_1(t),Y_1(t))_{t\geq 0}$ and
$(X_2(t),Y_2(t))_{t\geq 0}$ for the processes with parameters
$(p_k)_{k\geq0},\l,\g^{(1)}$ and $(p_k)_{k\geq0},\l,\g^{(2)}$,
respectively.  Let $X_1(0)=X_2(0)$ and $Y_1(0)=Y_2(0).$ There is no
interaction between $(X_1(t))_{t\geq 0}$ and $(Y_1(t))_{t\geq 0}$, and
$(Y_1(t))_{t\geq 0}$ simply forms a supercritical branching process
(this uses $\lambda<1$).

It is straightforward to couple 
$(Y_1(t))_{t\geq 0}$ with $(X_2(t),Y_2(t))_{t\geq 0}$ so that the
following hold.  Firstly, each transition $Y_1(t)\rightarrow Y_1(t)+1$
is accompanied by $(X_2(t),Y_2(t))\rightarrow(X_2(t),Y_2(t)+1)$.
This simply corresponds to the event that an infected
cell is replaced by two identical ones;  recall that $p_2=1$.
Secondly, if $X_2(t)\neq0$ then each transition $Y_1(t)\rightarrow
Y_1(t)-1$ is accompanied by
$(X_2(t),Y_2(t))\rightarrow(X_2(t)-1,Y_2(t))$.  
This corresponds to the death of an infected cell:
since $\g^{(1)}_0=1$ such an event simply reduces $Y_1(t)$
by $1$, and since $\g^{(2)}_1=1$ one healthy cell becomes
infected in the second process.
Thirdly, if $X_2(t)=0$
then each transition $Y_1(t)\rightarrow Y_1(t)-1$ is accompanied by
$(X_2(t),Y_2(t))\rightarrow(X_2(t),Y_2(t)-1)$.  
This corresponds to the death of an infected cell in the 
absence of healthy cells to infect.  In such a coupling, 
$Y_1(t)\leq Y_2(t)$ for all $t\geq0$, almost surely, so
$\eta(\g^{(2)})\leq\eta(\g^{(1)})$.  It is easy to see that the
inequality is in fact strict.

For the second part, we use the result in the forthcoming
article~\cite{BB} that the \emph{coexistence probability}
\[
\zeta(\l)=P(X(t)Y(t)\neq0\;\forall t\geq0)
\]
satisfies: $\zeta=0$ if $\b\geq \a>0$, and $\zeta>0$ if $\a>\b>0$.
Note that $\eta<1$ if $\zeta>0$.  If $\zeta(\l)=0$ and, in addition,
$\l\geq\a$, then it follows from part (2) of Proposition~\ref{br_prop}
that $\eta(\l)=1$: either $(Y(t))_{t\geq 0}$ becomes extinct before
$(X(t))_{t\geq 0}$, or $(X(t))_{t\geq 0}$ becomes extinct before
$(Y(t))_{t\geq 0}$, and in the latter case $(Y(t))_{t\geq 0}$
subsequently forms a branching process which has Malthusian parameter
$\b'=\a-\l\leq0$ and therefore becomes extinct almost surely.  It is easy
to check that if $p_0=3/8$, $p_2=5/8$, $E(\G)=4/5$, $\l_1=7/8$ and
$\l_2=17/8$, then $\beta(\l_1)=3/8>2/8=\a$ and $\l_1>\a$, whereas
$\b(\l_2)=1/8<\a$.  Thus $\zeta(\l_1)=0$ and $\eta(\l_1)=1$, but
$\zeta(\l_2)>0$ so $\eta(\l_2)<1$.
\end{proof}

\section{Discussion}
\label{disc_sec}

\subsection{Bacteriophage Lambda}

The virus Lambda, which preys on the bacterium
\emph{e-coli},
has been the subject of intensive research, mainly to
understand the fascinating lysis--lysogeny
behaviour~\cite{johnson81,oppenheim05}.
For this virus, a decision between lysis
and lysogeny occurs both \emph{at} the time of
infecting a new host, as well as after having been
incorporated in the host's DNA~\cite{lieb53}.
The switch
to lytic behaviour in response to stress to the host
seems inevitable~\cite{oppenheim05}.
In recent years, some exciting single-cell studies
have investigated the factors determining the
decision at the time of infection.  The results showed
strong dependence on environmental signals as well
as the volume of the infecting cell and
the number of infecting virions per
cell (or multiplicity of infection,~\cite{joh_weitz,pnas08,zeng10}.

A number of mathematical models have been proposed to
study the balance between the lytic and lysogenic
states~\cite{ackers82,arkin98,mcadams95,reinitz_vaisnys,santillan04,shea_ackers}.
With time
the models simulated more and more accurately by
including newly discovered genetic components, describing
a strict bias towards the lysogenic state  as
exhibited by `wet' experiments.  Both the experimental and
theoretical works revealed the molecular mechanism
of the decision, but did not study the
\emph{motivation} behind such a strict bias in the
decision system, as we do here.

\subsection{General conclusion}\label{gc_sec}
In this work we have tried to reveal the reason for
the observed stability of the lysogenic state.
Our results suggest that, regardless of the details
of the molecular mechanism behind the decision, the
stability of the dormant state is a fundamental
part of a long-term survival strategy.

To be specific, in Theorem~\ref{mono_thm} we showed, 
in the context of our model, not only that 
$\l=0$ maximizes the survival chance $1-\eta(\l)$
of the virus, but that $\eta(\l)$ is in fact
monotonically increasing in $\l\geq0$;
this was done under the assumption $\g_0=0$.
Heuristically, monotonicity in $\l$
holds for the following reason.  As long as 
$X(t)\neq0$, the infected process $(Y(t))_{t\geq0}$
behaves as a branching process with a higher
exponential growth rate than  $(X(t))_{t\geq0}$;
this follows from Proposition~\ref{br_prop},
equation~\eqref{beta_alpha}, and the fact that 
$\g_0=0$ implies $E(\G)\geq1$.  Typically,
therefore, $X(t)=0$ for some $t$,
after which point  $(Y(t+s))_{s\geq0}$
by Proposition~\ref{br_prop}
is a branching process with  Malthusian 
parameter $\b'=\alpha-\l$ which is decreasing
in $\l$.  The smaller $\l$ is, the larger should be the
chance that $(Y(t+s))_{s\geq0}$ survives.  

The intuition above is valid whenever $E(\G)\geq1$,
supporting the guess that the conclusion of 
Theorem~\ref{mono_thm} should hold whenever
$E(\G)\geq1$.  We have only been able to make the 
intuition rigorous when $\g_0=0$, essentially because
then $X(t)\rightarrow0$ almost surely (by 
Proposition~\ref{takeover_prop}), which does not hold if 
$\g_0>0$.  Interestingly, $\eta(\l)$ need not be 
monotone in $\l$ if $E(\G)<1$, as shown in 
Proposition~\ref{exprop};  this does not, however,
rule out the possibility that $\eta(\l)$
still always attains its minimum at $\l=0$.

\subsection{Decision at the time of infection}
\label{imm_sec}

The results above concern only the decision between
lysis and lysogeny \emph{after} the virus has
been incorporated in the host's DNA, and not decisions
\emph{at} the time of infection.  A simple version of a decision
at the time of infection can easily be incorporated
into our model as follows.

We modify transition rate (iii) in Section~\ref{def_sec}
so that, with a fixed probability $\k>0$, each newly infected
cell is \emph{immediately} replaced by a random (independent)
number $\G$ of infected cells,
taken from the healthy population $X(t)$.
This proceeds recursively for
all thereby newly infected cells, until there are either no
healthy cells left, or the recursion terminates by itself.
Note that the life-length
of an infected cell is now a convex combination of an exponential
distribution and a Dirac mass at 0, but that the process
$(X(t),Y(t))_{t \geq 0}$ is still Markovian.
As mentioned, the recursion terminates at the latest when the
healthy population $X(t)$ is exhausted. Therefore, there are no 
transitions of `infinite size'.

The total number of new infections
due to the original lysis event may be described using 
a random variable $\G'$,
whose distribution is easily described in terms of $\k,\G$ and $X(t).$
The process thus described is \emph{not}
simply the same as our main model with $\G$ replaced 
by $\G'$ (one easy way to see this is to note that
the new process can have transitions decreasing 
$X(t)+Y(t)$ by more than one at a time).  
However, it is possible to modify the 
proofs of Theorems~\ref{coupling_thm} and~\ref{mono_thm}
to obtain the following:
\begin{theorem}\label{immediate_thm}
In the process with decision at the time of infection
described above, with $\g_0=0$, the extinction
probability of $(Y(t))_{t \geq 0}$ is monotonically increasing in $\l$,
$(\g_k)_{k\geq0}$ and $\k$.
\end{theorem}

Briefly, the required modifications to 
Theorem~\ref{coupling_thm} are the following;  for
notation and terminology, see Section~\ref{proof_sec}.
In addition to the parameters $\l_1\leq\l_2$
and $\g^{(1)}\preceq\g^{(2)}$ we also have
$0\leq\k_1\leq\k_2\leq1$.  In the proof of 
Theorem~\ref{coupling_thm} there were several points
where pairs of colour $(R,R)$, $(B,R)$, $(R,\es)$
or $(R,G)$ were created.  These transitions are still
valid, but now, in addition, each newly created
$(R,R)$ will itself immediately undergo a type-1-lysis
with probability $\k_1$, or a type-2-lysis with
probability $\k_2-\k_1$.  Similarly, each new $(B,R)$
immediately undergoes a type-2-lysis with probability $\k_2$,
and each new $(R,\es)$ or $(R,G)$ immediately undergoes a 
lysis with probability $\k_1$.  The same is then done 
recursively for all thereby newly created pairs
$(R,R)$, $(B,R)$, $(R,\es)$ or $(R,G)$.  The order in which
these `immediate' transitions are carried out is
not important.  It is easy to see that the conclusions
of Theorem~\ref{coupling_thm} still hold under these modifications.

Theorem~\ref{immediate_thm} 
says that, when $\g_0=0$, the optimal `choice'
of $\l$ and $\k$ for the virus is $\l=\k=0$.
Returning to the bacteriophage Lambda,
which frequently lyses its host cell immediately
after infection, we conclude
that the model just described is inadequate as a description
of this virus.  The main confounding assumptions are presumably:
firstly, that $\g_0=0$;  secondly, that there is absolutely
no delay between the lytic phase and new infections; and thirdly,
that factors such as \textsc{moi}, which experiments have shown to be
important, are not included.  Furthermore, it is not hard to
imagine other factors which could make a more rapid increase
in numbers beneficial to the virus in the early stages of
an epidemic, such as competition from other viruses or an
immune response.  It is hoped that relevant
modifications of the model can be studied in future work.

\subsection{Future directions}
The main questions left open by this work are:
is Theorem~\ref{mono_thm} true whenever $E(\G)\geq1$? and,
what choices of $\l$ and $(\g_k)_{k\geq0}$ minimize $\eta$ when
$E(\G)<1$?

There are many natural ways to modify the model to make it
more realistic as a model for viruses.  One direction would be
to instead let $\G$ be the number of
new virions upon a lysis event, and let 
the number of new infections depends also on the ratio of $X(t)$
to  $Y(t)$.  Another direction would be to study the model with
decision at the moment of infection also when $\g_0>0$;
this requires some new arguments.

It is natural to consider the possibility of two competing
viruses, alternatively a virus competing with an immune system.
Finally, it would be natural to look at a version of the process
which is based not on branching process dynamics, but on the
dynamics of population models having some type of equilibrium 
like the logistic process~\cite{renshaw}.
Indeed it is reasonable to expect that the cell population
will be in equilibrium at the time of infection, possibly making
such a formulation closer to reality.

\bibliography{vip}
\bibliographystyle{plain}

\end{document}